\documentclass[journal,10pt,twocolumn]{IEEEtran}
\usepackage{epsfig}
\usepackage{graphicx}
\usepackage{subfigure}

\newtheorem{theorem}{Theorem}

\usepackage{amssymb}
\usepackage{makeidx}
\usepackage{amsmath}
\usepackage{graphicx}
\usepackage{epsf}
\usepackage{epsfig}
\usepackage{ccaption}
\usepackage{array}
\usepackage{tabularx}
\usepackage{multirow}
\usepackage{epsfig}
\usepackage{cite}

\begin{document}
\topmargin=-0.6in \oddsidemargin -0.5in \textwidth=7.4in
\textheight=9.5in

\title{\huge{Lifetime Improvement of Wireless Sensor Networks by Collaborative Beamforming and Cooperative Transmission}}
\author{Zhu Han$^*$ and H. Vincent Poor$^+$\\
$^*$Department of Electrical and Computer Engineering, Boise State
University, Idaho, USA\\
$^+$School of Engineering and Applied Science, Princeton University,
New Jersey, USA.\vspace{-10mm} \thanks{This research was supported
by the National Science Foundation under Grants ANI-03-38807 and
CCR-02-05214.}}

\maketitle

\begin{abstract}

Extending network lifetime of battery-operated devices is a key
design issue that allows uninterrupted information exchange among
distributive nodes in wireless sensor networks. Collaborative
beamforming (CB) and cooperative transmission (CT) have recently
emerged as new communication techniques that enable and
leverage effective resource sharing among
collaborative/cooperative nodes. In this paper, we seek to
maximize the lifetime of sensor networks by using the new idea
that closely located nodes can use CB/CT to reduce the load or
even avoid packet forwarding requests to nodes that have
critical battery life. First, we study the effectiveness of CB/CT
to improve the signal strength at a faraway destination using
energy in nearby nodes. Then, a 2D disk case is analyzed to assess the
resulting performance improvement. For general networks, if
information-generation rates are fixed, the new routing problem is
formulated as a linear programming problem; otherwise, the cost
for routing is dynamically adjusted according to the amount of
energy remaining and the effectiveness of CB/CT. From the analysis and
simulation results, it is seen that the proposed schemes can
improve the lifetime by about 90\% in the 2D disk network and by
about 10\% in the general networks, compared to existing
schemes.\vspace{-5mm}
\end{abstract}

\section{Introduction}\label{sec:intro}

In wireless sensor networks \cite{Akyildiz1}, extending the
lifetime of battery-operated devices is considered a key design
issue that increases the capability of uninterrupted information
exchange and alleviates the burden of replenishing batteries.  In
\cite{Tassiulas}, a data routing algorithm has been proposed with
an aim to maximize the minimum lifetime over all nodes in wireless
sensor networks. A survey of energy constraints for sensor networks
has been studied in \cite{Ephremides}. In \cite{Yates1}, the
network lifetime has been maximized by employing the accumulative
broadcast strategy. The work in \cite{Midkiff1} has considered
provisioning additional energy in the existing nodes and deploying
relays to extend the lifetime.

Recently, collaborative beamforming (CB) \cite{CB} and cooperative
transmission (CT) \cite{bib:Aazhang1}\cite{bib:Laneman2} have been
proposed  communication techniques that fully utilize spatial
diversity and multiuser diversity. While most existing work in this area concentrates
on improving the performances at the physical layer,
 CB and CT also have impact on the design of higher layer
protocols. In this paper, we investigate new routing protocols to
improve the lifetime of wireless sensor networks using these two
techniques.

First, we study the fact that CB/CT can effectively increase the
signal strength at a destination node, which in turn can increase the
transmission range. We obtain a closed-form analysis of the
effectiveness of CT similar to that given for CB in \cite{CB}. Then, we
formulate the problem as a maximization of the network lifetime, defined
until the time of the first node failure. The {\em new idea} is that closely located
nodes can use CB/CT to reduce the loads or even avoid packet
forwarding requests to nodes with critical battery lives.
From the analysis of a 2D disk case using CB/CT, we investigate
how  battery-depleting nodes close to the sink can be bypassed. Then
we propose algorithms for a general network situation. If the
information-generation rates are fixed, we can formulate the
problem as a linear programming problem. Otherwise, we propose a
heuristic algorithm to dynamically update costs in the routing
table according to the remaining energy and effectiveness of
collaboration. From the analysis and simulation results, the
proposed new routing schemes can improve the lifetime by about
90\% in the 2D disk network compared to the pure packet forwarding
scheme, and by about 10\% in general networks, compared to the
schemes in \cite{Tassiulas}.

This paper is organized as follows: In Section \ref{sec:model},
the system model is given, and the abilities of CB/CT to enhance the
destination signal strength are studied. In Section
\ref{sec:prob_def}, we formulate the lifetime maximization
problem, analyze a 2D disk case, and propose algorithms for
 general network situations. Simulation results are given in
Section \ref{sec:simulation} and conclusions are drawn in Section
\ref{sec:conclusion}.

\vspace{-3mm}
\section{System Model and Effectiveness of CB/CT}\label{sec:model}

We assume that a group of sensors is uniformly distributed with a density
of $\rho$. Each node is equipped with a single ideal isotropic
antenna. There is no power control for each node, i.e., the node
transmits with power either $P$ or $0$. There is no reflection or
scattering of the signal. Thus, there is no multipath fading or
shadowing. The nodes are sufficiently separated that any mutual
coupling effects among the antennas of different nodes are
negligible.

For traditional direct transmission, a node tries to reach another
node at a distance of $A$. The signal to noise ratio (SNR) is
given by\vspace{-1mm}
\begin{equation}
\Gamma=\frac {PC_0A^{-\alpha}|h|^2}{\sigma ^2},\vspace{-1mm}
\end{equation}
where $C_0$ is a constant that incorporates  effects such as
antenna gains, $\alpha$ is the propagation loss factor, $h$ is the
channel gain, and $\sigma^2$ is the thermal noise level. We define
the energy cost of such a transmission for each packet to be one
unit.

\begin{figure}[htbp]
\begin{center}
    \epsfig{file=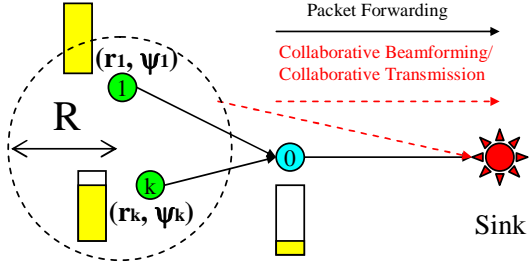,width=70truemm}
\end{center}
\caption{System Model}\label{system_model}\vspace{-5mm}
\end{figure}

In Figure \ref{system_model}, we show the system model with CB/CT.
In traditional sensor networks, the only choice a node has is to
forward its information toward the sink. This will deplete the
energy of the nodes near the sink, since they have to transmit many
 other nodes' packets. To overcome this problem, in this
paper, we propose another choice for a node consisting of forming CB/CT with
the nearby nodes so as to transmit further towards the sink. By
doing this, we can balance the energy usage of the nodes having
different locations and different remaining energy. In the rest of
this section, we study how effectively CB and CT can improve the
link quality. \vspace{-2mm}


\subsection{Effectiveness of Collaborative Beamforming}

Suppose there are a total of $N$ users for collaborative
beamforming within a disk of radius $R$. We have\vspace{-1mm}
\begin{equation}
N=\lfloor \rho \pi R^2\rfloor.\label{N_R}\vspace{-1mm}
\end{equation}
Each node has polar coordinate $(r_k,\psi_k)$ to the disk center.
The distance from the center to the beamforming destination is
$A$. The Euclidean distance between the $k^{th}$ node and the
beamforming destination can be written as:
\begin{equation}
d_k=\sqrt{A^2+r_k^2-2r_kA\cos(\phi-\psi_k)},
\end{equation}
where $\phi$ is azimuth direction and is assumed to be a constant.
By using loop control or the Global Positioning System, the initial
phase of node $k$ can be set to\vspace{-1mm}
\begin{equation}
\psi _k=-\frac {2\pi}{\lambda} d_k(\phi),\vspace{-1mm}
\end{equation}
where $\lambda$ is the wavelength of the radio frequency carrier.

Define $\textbf z=[z_1\dots z_N]^T$ with\vspace{-1mm}
\begin{equation}
z_k=\frac {r_k} R \sin ( \psi_k-\phi/2).
\end{equation}
The array factor of CB can be written as\vspace{-1mm}
\begin{equation}
F(\phi|\textbf z)=\frac 1 N \sum_{k=1}^N e^{-j4\pi R \sin(\frac
\phi 2)z_k/\lambda}.\vspace{-1mm}
\end{equation}
The far-field beam pattern can be defined as:\vspace{-1mm}
\begin{eqnarray}
P(\phi|\textbf z)&=&|F(\phi|\textbf z)|^2\\
&=&\frac 1 N +\frac 1{N^2} \sum_{k=1}^N e^{-ja(\phi)z_k}
\sum_{l\neq k} e^{ja(\phi)z_l},\nonumber\vspace{-1mm}
\end{eqnarray}
where\vspace{-1mm}
\begin{equation}
a(\phi)=\frac {4\pi R \sin \frac \phi 2}{\lambda}.
\end{equation}

Define the directional gain $D_{av}^{CB}$ as the ratio of radiated
concentrated energy in the desired direction over that of a single
isotropic antenna. From Theorem 1 in \cite{CB}, for large $\frac R
\lambda$ and $N$, the following lower bound for far-field
beamforming is tightly held:
\begin{equation}
    \frac{D_{av}^{CB}} N \geq \frac 1 {1+\mu\frac
    {N\lambda}{R}},\label{bound}
\end{equation}
where $\mu\approx 0.09332$. 


Considering this directional gain, we can improve the direct
transmission by a factor of $D_{av}^{CB}$. Notice that this
transmission distance gain for one transmission is obtained at the
expense of consuming a total power of $N$ units from the nearby
nodes. \vspace{-3mm}

\subsection{Effectiveness of Cooperative Transmission}

Similar to the CB case, we assume $N$ users are uniformly
distributed over a radius of $R$. The probability density function
of the users' radial coordinate $r$ is given by
\begin{equation}
q(r)=\frac {2r} {NR^2}, \ \ 0\leq r \leq R,
\end{equation}
and the users' angular coordinate $\psi$ is uniformly distributed between
$[0,2\pi)$.

Suppose at the first stage, node $1$ transmits to the next hop or
sink. Then in the following stages, node $2$ to node $N$ relay the
node $1$'s information if they decode it correctly. The received
signals at node $2$ to node $k$ at stage 1 can be expressed as:
\begin{equation}
z_k=\sqrt{P r_k^{-\alpha}}h_k^r x+n_k^r,\ k=2,\dots N,
\end{equation}
and the received signals at the destination in the following
stages are
\begin{equation}
y_k=\sqrt{P d_k^{-\alpha}}h_k x+n_k.
\end{equation}
Here $P$ is the transmitted power, $h_k^r$ and $h_k$ are the
channel gains of source-relay and relay-destination, which are
modeled as independent zero mean circularly symmetric complex
Gaussian random variables with unit variance, $x$ is the
transmitted data having unit power, and $n_k$ and $n_k^r$ are
independent thermal noises with noise variance $\sigma ^2$.

\begin{theorem}
Define $D_{av}^{CT}$ to be the energy enhancement at the
destination node due to CB. Under the far-field condition and the
assumption that channel links between source and relays are
sufficiently good, we have the following approximation:
\begin{equation}\label{approx_CC}
\frac {D_{av}^{CT}}{N}\approx \frac {1+(N-1)\ _2 F_1\left(\frac 2
\alpha, -L; \frac {\alpha+2}\alpha ; \frac
{\sigma^2R^\alpha}{4P}\right)}{N},
\end{equation}
where $L$ is the frame length and $_2 F_1$ is the Hypergeometric
function
\begin{equation}
_2 F _1 (a,b;c;z)=\sum_{n=0}^{\infty} \frac {(a)_n(b)_n}{(c)_n}
\frac {z^n}{n!},
\end{equation}
where $(a)_n=a(a+1)\cdots(a+n-1)$ is the Pochhammer symbol.
\end{theorem}
\begin{proof}
The SNR received by the $k^{th}$ user at stage one can be written
as
\begin{equation}
\Gamma_k=\frac {PC_0r_k^{-\alpha}|h_k^r|^2}{\sigma ^2},
\end{equation}
where $|h_k^r|^2$ is the magnitude square of the channel fade and
follows an exponential distribution with unit mean.

\begin{figure}[htbp]
\begin{center}
    \epsfig{file=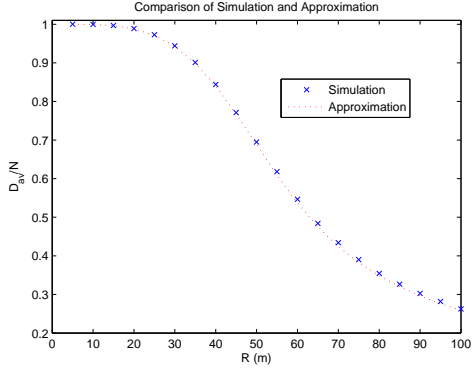,width=70truemm}
\end{center}
\caption{Approximation of CB
Effectiveness}\label{CC_compare}\vspace{-5mm}
\end{figure}

Without loss of generality, we suppose that BPSK modulation is used and
$C_0=1$. The probability of successful transmission of the packet with length $L$
is given by:\vspace{-1mm}
\begin{equation}\label{P_suc}
P_r^k(r)=\left( \frac 1 2 + \frac 1 2 \sqrt {\frac {P}{P+\sigma
^2r^{\alpha}_k}}\right) ^L.
\end{equation}

For fixed $(r_k,\psi_k)$, the average energy that arrives at the
destination can be written as:
\begin{equation}
D^{CT}=\sum_{k=1}^N P d_k^{-\alpha}P_r^k.
\end{equation}
Since for node $1$, $r_1=0$. We can write the average energy gain
in the following generalized form:
\begin{equation}\label{D_c_av_1}
D^{CT}_{av}=\sum_{k=1}^N \int _0 ^R \int _0 ^{2\pi} \frac
{A^\alpha} {2\pi} d_k^{-\alpha} P_r^k (r_k) q(r_k) dr_k d\psi _k.
\end{equation}

Since each user is independent of the others, we omit the notation
$k$ and can rewrite (\ref{D_c_av_1}) as:
\begin{eqnarray}
D^{CT}_{av}=1+(N-1)\\ \int _0 ^{2\pi} \int _0^R \frac {2r}
{A^{-\alpha}R^2} (A^2+r^2-2rA\cos\psi )^{-\frac \alpha 2} P_r(r)
dr d\psi \nonumber .
\end{eqnarray}

With the far field assumption, we have
\begin{equation}
\int _0 ^{2\pi} (A^2+r^2-2rA\cos \psi)^{-\frac \alpha 2} d\psi
\approx A^{-\alpha}.
\end{equation}
The average energy gain is approximated by
\begin{equation}
D^{CT}_{av}\approx 1+(N-1) \frac {2} {R^2} \int _0^R r P_r(r) dr .
\end{equation}
With the assumption of sufficiently good channels between sources
and relays, we have the following approximation of (\ref{P_suc}):
\begin{equation}
P_r(r)\approx (1-\frac {\sigma^2 r ^\alpha}{4P})^L.
\end{equation}
Since
\begin{equation}
\int _0^R r (1-\frac {\sigma^2 r ^\alpha}{4P})^L dr =\frac 1 2
R^2\ _2 F_1\left(\frac 2 \alpha, -L; \frac {\alpha+2}\alpha ;
\frac {\sigma^2R^\alpha}{4P}\right),
\end{equation}
we can obtain (\ref{approx_CC}).
\end{proof}

In Figure \ref{CC_compare}, we compare the numerical and
analytical results of $D_{av}$ for different radii $R$. Here
$A=1000$m, $P=10$dbm, $\sigma^2=-70$dbm, $\alpha=4$, $L=100$ and
$N=10$. We can see that the numerical result fits the analysis
very well, which suggests that the approximation in
(\ref{approx_CC}) is a good one.

\begin{figure}[htbp]
\begin{center}
    \epsfig{file=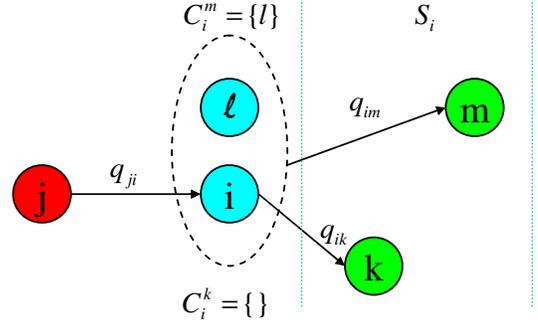,width=70truemm}
\end{center}
\caption{CB/CT Routing Model}\label{routing_mod}\vspace{-5mm}
\end{figure}

\section{CB/CT Lifetime Maximization\label{sec:prob_def}}

In this section, we first define the lifetime of sensor networks
and formulate the corresponding optimization problem. Then, by using a 2D disk
case, we demonstrate analytically the effectiveness of lifetime
saving using CB/CT. Finally, two algorithms are proposed for
general network configurations.

\subsection{Problem Formulation}

In Figure \ref{routing_mod}, we show the routing model with CB/CT.
A wireless sensor network can be modeled as a directed graph
$G(M,\mathbb{A})$, where $M$ is the set of all nodes and
$\mathbb{A}$ is the set of all links $(i,j), i,j\in M$. Here the
link can be either a direct transmission link or a link with
CB/CT. Let $S_i$ be the set of nodes that the $i^{th}$ node can
reach by direct transmission. Denote by $C_i^m$ the set of nodes
that node $i$ needs to apply CB/CT with in order to reach node $m$. In the example in
Figure \ref{routing_mod}, $C_i^m=\{ l\}$ and $C_i^k=\{ \}$. A set
of origin nodes $O$ where information is generated at node $i$
with rate $Q_i$ can be written as:
\begin{equation}
O=\{ i | Q_i>0,\ i\in M\}.
\end{equation}
A set of destination nodes is defined as $D$ where
\begin{equation}
D=\{ i | Q_i<0,\ i\in M\}.
\end{equation}

Define $\textbf q=\{q_{ij}\}$ to represent the routing and the
transmission rate. There are many types of definitions for
lifetime of sensor networks. The most common ones are the first
node failure, the average lifetime, and $\alpha$ lifetime. In this
paper, we use the lifetime until first node failure as an example.
Other types of lifetime can be examined in a similar way. Suppose
node $i$ has remaining energy of $E_i$. The lifetime for each node
can be written as:\vspace{-1mm}
\begin{equation}
T_i(\textbf q)=\frac {E_i}{\sum_{j\in S_i} q_{ij}+\sum _{i\in
C_j^m,\forall j,m} q_{jm}},\vspace{-1mm}
\end{equation}
where the first term in the denominator is for direct transmission
and the second term in the denominator is for CB/CT. Notice that
$C_j^m$ is not a function of $\textbf q$. We formulate the problem
as\vspace{-1mm}
\begin{equation}\label{Prob_def}
\max_{\textbf q}\ \min T_i
\end{equation}
\[
\mbox {s.t. }\left\{
\begin{array}{l}
q_{ij}\geq 0, \forall i,j\\
\sum_{j,i\in S_j} q_{ij} + Q_i=\sum_{k\in S_i}q_{ik},\\
\end{array}
\right.
\]
where the second constraint is for flow conservation.\vspace{-3mm}

\begin{figure}[htbp]
\begin{center}
    \epsfig{file=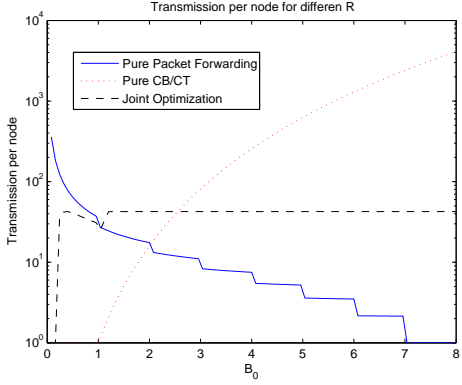,width=70truemm}
\end{center}
\caption{Analytical Results for the 2D Disk
Case}\label{analysis_fig}\vspace{-5mm}
\end{figure}

\subsection{2D Disk Case Analysis}\vspace{-1mm}

In this subsection, we study a 2D disk case network. Users with
the same remaining energy are uniformly located within a circle of
radius $B_0$. One sink is located at the center location $(0,0)$.
Each node has a unit amount of information to transmit. Here we
assume the user density is large enough, so that each node can
find enough nearby nodes to form CB/CT to reach the faraway node.

For traditional packet forwarding without CB/CT, the number of
packets needing transmission for each node at the distance $B$ to
the sink is given by\vspace{-2mm}
\begin{equation}
N_{pf}(B)=\sum_{n=0}^{\lfloor\frac {B_0-B}{A_0}\rfloor}
\left(1+\frac {nA_0}{B}\right).\vspace{-1mm}
\end{equation}
where $A_0$ is the maximal distance over which a minimal link
quality $\gamma_0$ can be maintained, i.e. $\Gamma(A_0)=\gamma_0$.

If all nodes use their neighbor nodes to communicate with the sink
directly, we call this scheme pure CB/CT. To achieve the range of
$B$, we need $N_{CB/CT}(B)$ for CB/CT, i.e.,\vspace{-1mm}
\begin{equation}
D_{av}\left(\sqrt{\frac
{N_{CB/CT}(B)}{\rho\pi}}\right)=\left(\max(\frac B
{A_0},1)\right)^\alpha.\vspace{-1mm}
\end{equation}

For collaborative beamforming, we can calculate\vspace{-1mm}
\begin{equation}
N_{CB}(B)\geq \frac 1 2
\left(c_0(2+c_0c_1^2)+c_0^{1.5}c_1\sqrt{4+c_0c_1^2}\right)\vspace{-1mm}
\end{equation}
where $c_0=(\max(\frac B {A_0},1))^\alpha$ and $c_1=\mu \lambda
\sqrt{\rho \pi}$. For cooperative transmission, numerical results
need to be used to obtain the inverse of $_2F_1$ in Theorem 1.
Notice that if the node density is large enough, then
$D_{av}/N\longrightarrow 1$.

In Figure \ref{analysis_fig}, we show the average transmission per
node vs. the disk size $B_0$. We can see that for traditional
packet forwarding, the node closest to the sink has the most
transmissions per node, i.e., it has the lowest lifetime if the
initial energy is the same for all nodes. On the other hand, for
the pure CB/CT scheme, more nodes need to transmit to reach the
sink directly when $B_0$ is larger. The transmission is less
efficient than packet forwarding, since the propagation loss
factor $\alpha$ is larger than $1$. The above facts motivate the
joint optimization case where nodes transmit packets with
different probabilities over traditional packet forwarding and
CB/CT.

For traditional packet forwarding, nodes near the sink have lower
lifetimes. If the faraway nodes can form CB/CT to transmit directly
to the sink and bypass these life depleting nodes, the overall
network lifetime can be improved. Notice that in this special
case, if the faraway nodes form CB/CT to transmit to nodes other
than the sink, the lifetime will not be improved. For each node
with distance $B$ to the sink, and supposing the probability of using
CB/CT is $P_r(B)$, we have

\begin{eqnarray}\label{N_joint}
N_{joint}(B)= (1-P_r(B)+N_{CB/CT}(B)P_r(B))\nonumber
\\\sum_{n=0}^{\lfloor\frac {B_0-B}{A_0}\rfloor}\left(1+\frac {
nA_0}{B}\right)\Pi _{j=1}^n(1-P_r(B+jA_0)),\vspace{-1mm}
\end{eqnarray}
where the first term on the right-hand side (RHS) is the necessary
energy for transmitting one packet, and the second term is the
number of packets for transmission. The goal is to adjust $P_r(B)$
such that the lifetime is maximized, i.e.,
\begin{equation}\label{bisection_prob}
\min_{1 \geq P_r(B)\geq 0} \max N_{joint}(B).\vspace{-1mm}
\end{equation}

Notice that $N_{CB/CT}\geq 1$, and in (\ref{N_joint}) the second
term on the RHS depends on the probabilities of CB/CT being larger
than $B$. So we can develop an efficient bisection search method
to calculate (\ref{bisection_prob}). We define a temperature
$\kappa$ that is assumed to be equal or greater than
$N_{joint}(B),\forall B$. We can first calculate $N_{joint}(B)$
from the boundary of the network where the second term on the RHS
of (\ref{N_joint}) is one. Then we can derive all $N_{joint}(B)$
by reducing $B$. If $\kappa$ is too large, most information is
transmitted by CB/CT, and the nodes faraway from the sink waste
too much power for CB/CT; on the other hand, if $\kappa$ is too
small, the nodes close to the sink must forward too many packets.
A bisection search method can find the optimal values of $\kappa$
and $N_{joint}(B)$.

\begin{table}
\caption{Lifetime Saving vs. Disk Size}
\begin{center}
\begin{tabular}{|c|c|c|c|c|c|}
  \hline
  $R_0$ & 2 & 4 & 6 & 8 & 10 \\
  \hline
  $\max N_{joint}(B)$ & 2.82 & 10.25 & 23.4 & 42.5 & 64.5 \\

  \hline

  Saving \% & 94.56 & 93.33 & 90.86 & 88.13 & 85.98 \\
  \hline
\end{tabular}
\end{center}\label{ana_table}\vspace{-10mm}
\end{table}

In Figure \ref{analysis_fig}, we show the joint optimization case
where the node density is sufficiently large. We can see that to
reduce the packet forwarding burdens of the nodes near the sink,
the faraway nodes form CB/CT to transmit to the sink directly.
This will increase the number of transmissions per node for them,
but reduce the transmissions per node for the nodes near the sink.
In Table \ref{ana_table}, we show the maximal $N_{joint}(B)$ and
the lifetime saving over the traditional packet forwarding. We can
see that the power saving is around 90\%.\vspace{-3mm}

\subsection{General Case Algorithms\label{sec:algorithm}}

In this section, we first consider the case in which the
information generation rates are fixed for all sensors, and
develop a linear programming method to calculate the routing
table. Here to simplify the calculation of set $C_j^m$, we assume
its size equals  one. Obviously, this is suboptimal for
(\ref{Prob_def}). Then we select the nearest neighbor for CB/CT.
$C_j^m=1$, if node $j$'s nearest neighbor can help node $j$ to
reach node $m$. Define $\hat q_{ij}=Tq_{ij}$. The problem can be
written as a linear programming problem:\vspace{-1mm}
\begin{equation}
\max \ T
\end{equation}
\[
\mbox{s.t. } \left\{
\begin{array}{l}
\hat q_{ij}\geq 0, \forall i,j;\\
(\sum_{j\in S_i}\hat q_{ij}+\sum _{i\in C_j^m,\forall j,m}\hat q_{jm})\leq E_i, \ \forall i;\\
\sum_{j,i\in S_j}\hat q_{ji} + TQ_i=\sum_{k\in S_i}\hat
q_{ik},\forall i \in M-D,
\end{array}
\right.\vspace{-1mm}
\]
where the second constraint is the energy constraint and the third
constraint is for flow conservation.

\begin{figure}[htbp]
\begin{center}
    \epsfig{file=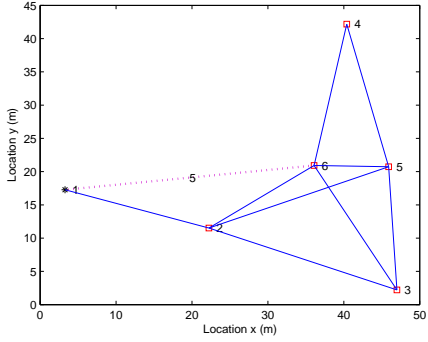,width=65truemm}
\end{center}
\caption{Snapshot of CB/CT
Routing}\label{sim_example}\vspace{-5mm}
\end{figure}

Next, if the information rate is random, each sensor dynamically
updates its cost according to its remaining energy and with
consideration of CB/CT. Some heuristic algorithms can be proposed
to update the link cost dynamically. Here the initial energy is
$E_i$. Define the current remaining energy as $\underline{E}_i$.
We define the cost for node $i$ to communicate with node $j$ as

\begin{equation}\label{link_cost}
\mbox{cost}_{ij}=\left(\frac{E_i}{\underline{E}_i}\right)^{\beta_1}+
\sum _{l\in C_i^j} \left(
\frac{E_l}{\underline{E}_l}\right)^{\beta_2},
\end{equation}
where $\beta_1$ and $\beta_2$ are positive constants. Their values
determine how the packets are allocated between the energy
sufficient and energy depleting nodes, and between the direct
transmission and CB/CT. Notice that (\ref{link_cost}) can be
viewed as an inverse barrier function for $\underline{E}_i\geq 0$.

%
%
%
%
%
%
%
%
%
%
%

\vspace{-3mm}

\section{Simulation Results}\label{sec:simulation}\vspace{-1mm}

We assume nodes and one sink are randomly located within a square
of size $\mathbb{L}\times \mathbb{L}$. Each node has power of
10dbm and the noise level is -70dbm. The propagation loss factor
is $4$. The minimal link SNR is 10dB. The initial energy of all
users is assumed to be unit and information rates for all users
are 1.

In Figure \ref{sim_example}, we show a snapshot of a network of
$5$ sensor nodes and a sink with $\mathbb{L}=50$m. Here node $1$
is the sink. The solid lines are the links for the direct
transmission, and the dotted line from node $6$ to the sink is the
CB/CT link with the help of node $5$. For traditional direct
packet forwarding scheme, the best flow is\vspace{-1mm}
\begin{equation}
\hat q_{ij}=\left[
\begin{array}{cccccc}
  0 & 0 & 0 & 0 & 0 & 0 \\
  1 & 0 & 0 & 0 & 0 & 0 \\
  0 & 0.2 & 0 & 0 & 0 & 0 \\
  0 & 0 & 0 & 0 & 0.1 & 0.1 \\
  0 & 0.3 & 0 & 0 & 0 & 0 \\
  0 & 0.3 & 0 & 0 & 0 & 0 \\
\end{array}
\right]\vspace{-1mm}
\end{equation}
with the resulting energy consumed for all nodes given by $[0,
1.0,0.2,0.2,0.3,0.3]$. Because node $2$ is the only node that can
communicate with the sink, the best lifetime of this routing is
$0.2$ before node $2$ runs out of energy.
\begin{figure}[htbp]
\begin{center}
    \epsfig{file=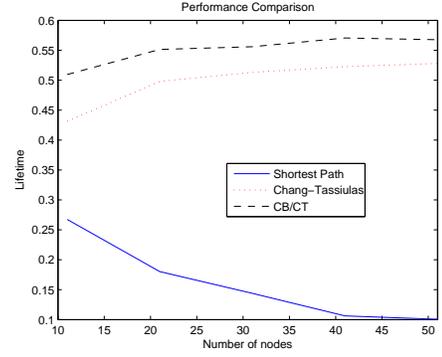,width=65truemm}
\end{center}
\caption{Lifetime Improvement Comparison}
\label{life_no_user}\vspace{-5mm}
\end{figure}

With CB/CT, the best flow is

\begin{equation}
\hat q_{ij}=\left[
\begin{array}{cccccc}
  0 & 0 & 0 & 0 & 0 & 0 \\
  1 & 0 & 0 & 0 & 0 & 0 \\
  0 & 0.321 & 0 & 0 & 0 & 0.012 \\
  0 & 0 & 0 & 0 & 0 & 0.333 \\
  0 & 0.23 & 0 & 0 & 0 & 0.103 \\
  0.667 & 0.115 & 0 & 0 & 0 & 0 \\
\end{array}
\right]
\end{equation}
with the energy consumed for all nodes given by $[0, 1.000, 0.333,
0.333, 1.0000, 0.782]$. Here some flow can be sent to the sink via
node $6$. Because of CB/CT, node $5$ has to consume its power. The
lifetime becomes $0.333$ which is 67\% improvement over direct
packet forwarding.

In Figure \ref{life_no_user}, we compare the performance of three
algorithms, the shortest path, the algorithm in \cite{Tassiulas},
and the proposed CB/CT algorithm. Here $\mathbb{L}=100$m. As the
number of users increases, the performance of the shortest path
algorithm decreases. This is because more users will need packet
forwarding by the nodes near the sink. Consequently, they die more
quickly. Compared with the algorithm in \cite{Tassiulas}, the
proposed schemes have about 10\% performance improvement. This is
because of the alternative routes to the sink that can be found by
CB/CT.\vspace{-3mm}

%
%
%

\section{Conclusions}\label{sec:conclusion}

In this paper, we have studied the impact of CB/CT on the design
of higher level routing protocols. Specifically, using CB/CT, we
have proposed a new idea based on bypassing energy depleting nodes
that might otherwise forward packets to the sink, in order to
improve the lifetime of wireless sensor networks. From the
analytical and simulation results, we have seen that the proposed
protocols can increase lifetime by about 90\% in a 2D disk case
and about 10\% in  general network situations, compared with
existing techniques.\vspace{-3mm}

\bibliographystyle{IEEE}

\end{document}